\documentclass{llncs}

\usepackage{makeidx}  

\usepackage{xcolor}
\usepackage{hyperref}
\usepackage{amsmath, amssymb}
\usepackage{epsfig,psfrag}
\usepackage{enumerate}
\usepackage{nicefrac}

\newcommand{\rn}[1]{\textcolor{black}{#1}}

\newcommand{\Z}{\mbox{\rm\bf Z}}
\newcommand{\Zplus}{\Z^+}

\date{}
\title{An Incentive Compatible, Efficient Market for Air Traffic Flow Management}
\titlerunning{An Incentive Compatible, Efficient Market for ATFM}

\author{Ruta Mehta\inst{1} \and Vijay V. Vazirani\inst{2}}
\institute{ Dept of Computer Science, University of Illinois at Urbana-Champaign
\and College of Computing, Georgia Tech}
\authorrunning{R. Mehta and V. V. Vazirani} 
\begin{document}

\maketitle

\begin{abstract}
We present a market-based approach to the Air Traffic Flow Management (ATFM) problem.  The goods in our market are delays and buyers
are airline companies; the latter pay money to the Federal Aviation Administration (FAA) to buy away the desired amount of delay on a
per flight basis. We give a notion of equilibrium for this market and an LP whose every optimal solution gives an equilibrium allocation of
flights to landing slots as well as equilibrium prices for the landing slots.  Via a reduction to matching, we show that this
equilibrium can be computed combinatorially in strongly polynomial time.  Moreover, there is a special set of equilibrium prices, which
can be computed easily, that is identical to the VCG solution, and therefore the market is incentive compatible in dominant strategy.  
 \end{abstract}

\section{Introduction}
Air Traffic Flow Management (ATFM) is a challenging operations research problem whose importance keeps escalating with
the growth of the airline industry. In the presence of inclement weather, the problem becomes particularly serious and
leads to substantial monetary losses and delays\footnote{According to \cite{Bertsimas2}, the U.S. Congress Joint Economic Committee
estimated that in 2007, the loss to the U.S. economy was \$25.7 billion, due to 2.75 million hours of flight delays. In
contrast, the total profit of U.S. airlines in that year was \$5 billion.},  Yet, despite massive efforts on the part of
the U.S. Federal Aviation Administration (FAA), airline companies, and even the academia, the problem remains largely unsolved. 

In a nutshell, the reason for this is that any viable solution needs to satisfy several conflicting requirements, e.g., in
addition to ensuring efficiency the solution also needs to be viewed as ``fair'' by all parties involved. Indeed,
\cite{Bertsimas3} state that `` ... While this work points at the possibility of dramatically reducing delay costs to the
airline industry vis-a-vis current practice, the vast majority of these proposals remain unimplemented. The ostensible
reason for this is fairness ... .'' It also needs to be computationally efficient -- even moderate sized airports today
handle hundreds of flights per day, with the 30 busiest ones handling anywhere from 1000 to 3000 flights per day.  The full
problem involves scheduling flight-landings simultaneously for multiple airports over a large period of time, taking into
consideration inter-airport constraints. Yet, according to \cite{Bertsimas2}, current research has 
mostly remained at the level of a single airport because of computational tractability reasons.

Building on a sequence of recent ideas that were steeped in sound economic theory, and drawing on ideas from game theory and
the theory of algorithms, we present a solution that has a number of desirable properties. Our solution for allocating
flights to landing slots at a single airport is based on the principle of a free market, which is known to be fair and a remarkably
efficient method for allocating scarce resources among alternative uses (sometimes stated in the colorful language of the
``invisible hand of the market'' \cite{smith}). We define a market in which goods are delays and buyers are airline companies; 
the latter pay money to the FAA to buy away the desired amount of delay on a per flight basis and we give a notion of equilibrium 
for this market. W.r.t. equilibrium prices, the total cost (price paid and cost of delay) of each agent, i.e., flight, is minimized.

This involves a multi-objective optimization, one for each agent, just like all market equilibrium problems. Yet, for some markets 
an equilibrium can be found by optimizing only one function. As an example, consider the linear case of Fisher's market \cite{scarf}
for which an optimal solution to the remarkable Eisenberg-Gale \cite{eisenberg} convex formulation gives equilibrium prices
and allocations. For our market, we give a special LP whose optimal solution gives an equilibrium. 

Using results from matching theory, we show how to find equilibrium allocations and prices in
strongly polynomial time.  Moreover, using \cite{leonard} it turns out that our solution is incentive compatible in dominant strategy,
i.e., the players will not be able to game the final allocation to their advantage by misreporting their private information.  

We note that the ATFM problem involves several issues that are not of a game-theoretic or algorithmic nature, e.g.,  the
relationship between long term access rights (slot ownership or leasing) and short term access rights on a given day of
operations, e.g., see \cite{ball1}.  Our intention in this paper is not to address the myriad of such issues. Instead, we have
attempted to identify a mathematically clean, core problem that is amenable to the powerful tools
developed in the theories stated above, and whose solution could form the core around which a practical scheme can be built.  

Within academia, research on this problem started with the pioneering work of Odoni \cite{Odoni} and it flourished with the
extensive work of Bertsimas et. al. ; we refer the reader to \cite{Bertsimas2}\cite{Bertsimas1} for thorough literature overviews
and references to important papers.  These were centralized solutions in which the FAA decides a schedule that is efficient,
e.g., it decides which flights most critically need to be served first in order to minimize cascading delays in the entire
system. 

A conceptual breakthrough came with the realization that {\em the airlines themselves are the best judge of how to achieve
efficiency\footnote{e.g., they know best if a certain flight needs to be served first because it is carrying CEOs of
important companies who have paid a premium in order to reach their destination on time or if delaying a certain flight by
30 minutes will not have dire consequences, however delaying it longer would propagate delays through their entire system
and result in a huge loss.},} thus moving away from centralized solutions. This observation led to a solutions based on
collaborative decision making (CDM) which is used in practice \cite{ball2,ball3,wambsganss}. 

More recently, a market based approach was proposed by Castelli, Pesenti and Ranieri \cite{RC}. Although their formulation
is somewhat complicated, the strength of their approach lies in that it not only
leads to efficiency but at the same time, it finesses away the sticky issue of fairness -- whoever pays gets smaller
delays, much the same way as whoever pays gets to fly comfortably in Business Class!  Paper \cite{RC} also gave a
tatonnement-based implementation of their market. Each iteration starts with FAA announcing prices for landing slots. Then,
airlines pick   their most preferred slots followed by FAA adjusting prices, to bring parity between supply and demand, for
the next iteration. However, they are unable to show convergence of this process and instead propose running it a
pre-specified number of times, and in case of failure, resorting to FAA's usual solution. They also give an example for which
incentive compatibility does not hold. 

Our market formulation is quite different and achieves both efficient running time and incentive compatibility.
We believe that the simplicity of our solution for this important problem, and the fact that it draws on fundamental ideas from 
combinatorial optimization and game theory, should be viewed as a strength rather than a weakness.

\subsection{Salient features of our solution}
\label{sec.salient}

In Section \ref{sec.model} we give details of our basic market model for allocating a set of flights to landing slots for one airport. 
This set of flights is picked in such a way that their actual arrival times lie in a window of a couple of hours; the reason for the
latter will be clarified in Section \ref{sec.multiple}.  The goods in our market are delays and buyers are airline
companies; the latter pay money to the FAA to buy away the desired amount of delay on a per flight basis.  Typically flights
have a myriad interdependencies with other flights -- because of the use of the same aircraft for subsequent flights,
passengers connecting with other flights, crew connecting with other flights, etc. The airline companies, and not FAA, are
keenly aware of these and are therefore in a better position to decide which flights to avoid delay for. The information
provided by airline companies for each flight is the dollar value of delay as perceived by them. 

For finding equilibrium allocations and prices in our market, we give a special LP in which parameters can be set according
to the prevailing conditions at the airport and the delay costs declared by airline companies. We arrived at this LP as
follows. Consider a traffic network in which users selfishly choose paths from their source to destination. One way of
avoiding congestion is to impose tolls on roads. \cite{CDR} showed the existence of such tolls for minimizing the total
delay for the very special case of one source and one destination, using Kakutani's fixed point theorem. Clearly, their
result was highly non-constructive. In a followup work, \cite{FJM} gave a remarkable LP whose optimal solution yields such
tolls for the problem of arbitrary sources and destinations and moreover, this results in a polynomial time algorithm. Their
LP, which was meant for a multi-commodity flow setting, was the starting point of our work. One essential difference between
the two settings is that whereas they sought a Nash equilibrium, we seek a market equilibrium; in particular, the latter
requires the condition of market clearing.  

We observe that the underlying matrix of our LP is totally unimodular and hence it admits an integral optimal solution. Such
a solution yields an equilibrium schedule for the set of flights under consideration and the dual of this LP yields
equilibrium price for each landing slot. Equilibrium entails that each flight is scheduled in such a way that the sum of the
delay price and landing price is minimum possible. We further show that an equilibrium can be found via an algorithm for the
minimum weight perfect $b$-matching problem and hence can be computed combinatorially in strongly polynomial time. In
hindsight, our LP resembles the $b$-matching LP, but there are some differences.  

Since the $b$-matching problem reduces to the maximum matching problem, our market is essentially a matching market.
Leonard \cite{leonard} showed that the set of equilibrium prices of a matching market \rn{with minimum sum} 
corresponds precisely to VCG payments \cite{nisan.chap9},  thereby showing that the market is incentive compatible
in dominant strategy. \rn{Since equilibrium prices form a lattice \cite{DG,SS,AJM}, the one minimizing sum has to be simultaneously minimum for all goods.}
For our market, we give a simple, \rn{linear-time} procedure that converts arbitrary equilibrium prices to ones that are
simultaneously minimum for all slots. Incentive compatibility with these prices follows.  An issue worth mentioning is that
the total revenue, or the total cost, of VCG-based incentive compatible mechanisms has been studied extensively, mostly with
negative results \cite{AT,KKT,ESS,CS,HT}. In contrast, since the prices in our natural market model happened to be VCG
prices, we have no overhead for making our mechanism incentive compatible.

The next question is how to address the scheduling of landing slots over longer periods at multiple airports, taking into
consideration inter-airport constraints. Airlines can and do anticipate future congestion and delay issues and take these into
consideration to make complex decisions. However, sometimes  unexpected events happening even at a few places are likely to
have profound cascading effects at geographically distant airports, making it necessary to make changes dynamically.
For such situations, in Section \ref{sec.multiple}, we propose a
dynamic solution by decomposing this entire problem into many small problems, each of which will be solved by the method
proposed above. The key to this decomposition is the robustness of our solution for a single set of flights at one airport:
we have not imposed any constraints on delay costs, not even monotonicity. Therefore, airline companies can utilize this
flexibility to encode a wide variety of inter-airport constraints.

We note that this approach opens up the possibility of making diverse types of travelers happy through the following
mechanism: the additional revenues generated by FAA via our market gives it the ability to subsidize landing fees for low
budget airlines. As a result, both types of travelers can achieve an end that is most desirable to them, business travelers
and casual/vacation travelers. The former, in inclement weather, will not be made to suffer delays that ruin their important
meetings and latter will get to fly for a lower price (and perhaps sip coffee for an additional hour on the tarmac, in
inclement weather, while thinking about their upcoming vacation).  

To the best of our knowledge, ours is the first work to give a simple LP-based efficient solution for the ATFM problem. 
We note that an LP similar to ours is also given in \cite{AJM}. This paper considers two-sided matching markets with payments and non-quasilinear utilities. 
They show that the lowest priced competitive equilibria are group strategy proof, which induces VCG payments for the case of quasilinear utilities.
Another related paper is \cite{NDG}, which considers a Shapley-Shubik assignment model for unit-demand buyers and sellers with one indivisible item each. 
Buyers have budget constraint for every item. This sometimes prevents a competitive equilibrium from existing. They give a strongly polynomial-time 
algorithm to check if an equilibrium exists or not, and if it does exist, then it computes the one with lowest prices.
However, they do not ensure incentive compatibility.

\vspace{-0.2cm}

\section{The Market Model}
\label{sec.model}
In this section we will consider the problem of scheduling landings at one airport only.  Let $A$ be the set of all flights,
operated by various airlines, that land in this airport in a given period of time.  We assume that the given period of time
is partitioned into a set of landing time slots, in a manner that is most convenient for this airport; let $S$ denote this
set. Each slot $s$ has a capacity $cap(s) \in \Zplus$ specifying the number of flights that can land in this time slot. As
mentioned in \cite{ball1} the arrival of each aircraft consumes approximately the same amount of airport capacity, therefore
justifying the slot capacities as the number of flights while ignoring their types.  We will assume that $cap(s)$ is
adjusted according to the prevailing weather condition. 

For $i \in A$, the airline of this flight decides the {\em landing window} for flight $i$, denoted by $W(i)$. This gives the
set of time slots in which this flight should land as per prevailing conditions, e.g., if there are no delays, the earliest
time slot in $W(i)$ will be the scheduled arrival time\footnote{We will assume that if the flight arrives before this time, it will have to
wait on the tarmac for some time. This appears to be standard practice in case gates are not available.} of flight $i$.
For each slot $s \in W(i)$, the airline also decides
its {\em delay cost}, denoted by $c_{is}\ge 0$. Thus, if time slot $s$ is the scheduled arrival time of flight $i$, then
$c_{is} = 0$\footnote{All the results of this paper hold even if $c_{is} \neq 0$.} and in general $c_{is}$ is the dollar
value of the cost, as perceived by the airline, for delay resulting from landing in slot $s$.

A {\em landing schedule} is an assignment of flights to time slots, respecting capacity constraints. Each time slot will be
assigned a {\em landing price} which is the amount charged by FAA from the airline company if its flight lands in this
time slot. We will define the {\em total cost} incurred by a flight to be the sum of the price paid for landing and the
cost of the delay.

We say that a given schedule and prices are {\em an equilibrium landing schedule and prices} if:
\begin{enumerate}
\item
W.r.t. these prices, each flight incurs a minimum total cost.
\item
The landing price of any time slot that is not filled to capacity is zero. This condition is justified by observing that if at equilibrium,
a slot with zero price is not filled to capacity, then clearly its price cannot be made positive.
This is a standard condition in equilibrium economics.
\end{enumerate}

\subsection{LP formulation}
\label{sec.LP}

In this section, we will give an LP that yields an equilibrium schedule; its dual will yield equilibrium landing prices. 
Section \ref{sec.rounding} shows how they can be computed in strongly polynomial time.

For $s \in S$, $x_{is}$ will be the indicator variable that indicates whether flight $i$ is scheduled in time slot $s$;
naturally, in the LP formulation, this variable will be allowed to take fractional values. The LP given below obtains a
scheduling where a flight may be assigned partially to a slot (fractional scheduling), that minimizes the total dollar value of the delays incurred by all flights, subject to capacity constraints of the time slots. (Note that the inequality in the first constraint will be satisfied with equality
since the objective is being minimized; the formulation below was chosen for reasons of convenience).

\begin{equation}
\label{LP}
\begin{array}{ll}
\mbox{minimize} & \sum_{i \in A,s\in S} c_{is} x_{is} \\ 
\mbox{subject to}       & \forall i \in A: \ \sum_{s \in W(i)}  x_{is} \geq 1   \\
          & \forall s \in S:  \ \sum_{i \in A, s \in W(i)} x_{is}  \leq cap(s)     \\
          & \forall i \in A, \  s \in W(i):  \  x_{is}  \geq 0     

\end{array}
\end{equation}

Let $p_s$ denote the dual variable corresponding to the second set of inequalities. We will interpret $p_s$ as the
price of landing in time slot $s$. Thus if flight $i$ lands in time slot $s$,
the total cost incurred by it is $p_s + c_{is}$.
Let $t_i$ denote the dual variable corresponding to the first set of inequalities. In Lemma \ref{lem.cost} we will prove
that $t_i$ is the total cost incurred by flight $i$ w.r.t. the prices found by the dual; moreover, each flight incurs
minimum total cost.
  
The dual LP is the following.

\begin{equation}
\label{dual}
\begin{array}{ll}
\mbox{maximize} & \sum_{i \in A} t_i \  -  \ \sum_{s \in S} cap(s) \cdot p_{s}  \\
\mbox{subject to}       & \forall i \in A, \ \forall s \in W(i):  \  t_{i}  \leq   p_s \  +   \  c_{is}       \\
          & \forall i \in A: \  t_{i} \geq 0  \\
          & \forall s \in S:  \   p_{s}  \geq  0    
\end{array}
\end{equation}

\begin{lemma}
\label{lem.cost}
W.r.t. the prices found by the dual LP (\ref{dual}), each flight $i$ incurs minimum total cost and it is given by $t_i$.
\end{lemma}

\begin{proof}
Applying complementary slackness conditions to the primal variables we get
\[ \forall i \in A, \ \forall s \in W(i):  \ \ x_{is} > 0 \ \  \Rightarrow \ \
t_{i}  =  p_s \  +   \  c_{is}   . \]
Moreover, for time slots $s \in S$ which are not used by flight $i$, i.e., for which $x_{is} = 0$, by the dual
constraint, the total cost of using this slot can only be higher than $t_i$. The lemma follows.
\end{proof}

The second condition required for equilibrium is satisfied because of complementarity applied to the variables $p_s$:
\[ \mbox{If} \ \ \sum_{i\in A, s\in W(i)}   x_{is}  < cap(s), \ \ \mbox{then} \ \ p_s = 0 .\]

At this point, we can provide an intuitive understanding of how the actual slot assigned
to flight $i$ by LP (\ref{LP}) is influenced by the delay costs declared for flight $i$ and how LP (\ref{dual}) sets prices of slots.
Assume that time slot $s$ is the scheduled arrival time of flight $i$, i.e., $c_{is} = 0$ and $s'$ is a later slot. Then by Lemma \ref{lem.cost},
slot $s$ will be preferred to slot $s'$ only if $p_s - p_{s'} \leq c_{is'}$. Thus $c_{is'}$ places an upper bound on the extra money 
that can be charged for buying away the delay incurred by landing in $s$ instead of $s'$. Clearly, flight $i$ will 
incur a smaller delay, at the cost of paying more, if its airline declares large delay costs for late landing. Furthermore,
by standard LP theory, the dual variables, $p_s$, will adjust
according to the demand of each time slot, i.e., a time slot $s$ that is demanded by a large number of flights that have declared large delay costs 
will have a high price. In particular, if a slot is not allocated to capacity, its price will be zero as shown above.

It is easy to see that the matrix underlying LP (\ref{LP}) is totally unimodular. Therefore, it has an integral optimal solution.
Further, minimization ensures that for every flight $i$ at most one of the $x_{is}$s is one and the rest are zero. Hence we get:

\begin{theorem}\label{thm.equi}
Solution of LP (\ref{LP}) and its dual (\ref{dual}) give an (optimal) equilibrium schedule and equilibrium prices. 
\end{theorem}

\section{Strongly Polynomial Implementation}
\label{sec.rounding}

As discussed in the previous section, LP (\ref{LP}) has an integral optimal solution as its underlying matrix is totally unimodular.
In this section, we show that the problem of obtaining such a solution can be reduced to a minimum
weight perfect $b$-matching problem\footnote{The instance we construct can also be reduced to a minimum weight perfect
matching problem with quadratic increase in number of nodes.}, and hence can be found in strongly polynomial time; see
\cite{schrijver.book} Volume A. The equilibrium prices, i.e., solution of (\ref{dual}), can be obtained from the
dual variables of the matching. Furthermore, we show that there exist equilibrium prices that induce VCG payments, and
hence is incentive compatible in dominant strategy. Finally, we give a strongly polynomial time procedure to
compute such prices.

Consider the edge-weighted bipartite graph $(A', S,E)$, with bipartition $A' = A \cup \{v\}$, where $A$ is the set of 
flights and $v$ is a special vertex, and $S$ is the set of time slots. The set of edges $E$ and weights are as
follows: for $i \in A, \ s \in W(i)$, $(i, s)$ is an edge with weight $c_{is}$, and for each $s \in S,$ there
are $cap(s)$ many $(v, s)$ edges\footnote{This is not going to affect strong polynomiality, because we can assume that
$cap(s)\le |A|, \forall s$ without loss of generality.}, each with unit weight (a multi-graph).

The matching requirements are: $b_i=1$ for each $i \in A$, $b_s=cap(s)$ for each $s \in S$, and $b_v$$=\sum_{s \in S} {cap(s)} -
|A|$ for $v$. Clearly, the last quantity is non-negative, or else LP (\ref{LP}) is infeasible. 
The following lemmas show that the equilibrium landing schedule and prices can be computed using minimum weight perfect
$b$-matching of graph $(A', S, E)$.

\begin{lemma}\label{le.matching}
Let $F^*\subset E$ be a perfect $b$-matching in $(A',S,E)$ and $x^*$ be a schedule where $x^*_{is}=1$ if $(i,s)\in F^*$. $F^*$
is a minimum weight perfect $b$-matching if and only if $x^*$ is an optimal solution of LP (\ref{LP}).
\end{lemma}
\begin{proof}
To the contrary suppose $x'$ and not $x^*$ is the optimal solution of LP (\ref{LP}). Let
$F'=\{(i,s)\in E\ |\ x'_{is}=1\}\cup \{(cap(s) - \sum_{i; s \in W(i)} x'_{is}) \mbox{ many } (v,s)\ |\ s \in S\}$ 
be the set
of edges corresponding to schedule $x'$. Clearly, $F'$ is a perfect $b$-matching. Note that the matching edges
incident on $v$ contribute cost $b_v$ in any perfect $b$-matching. Since, $x'$ and not $x^*$ is an optimal solution of 
LP (\ref{LP}), we have, 
\[
\begin{array}{lcl}
\displaystyle\sum_{i \in A, s \in W(i)} c_{is} x'_{is} + b_v <
\displaystyle\sum_{i \in A, s \in W(i)} c_{is} x^*_{is} + b_v
& \Rightarrow & \displaystyle\sum_{(i,j)\in F'} c_{ij} < \displaystyle\sum_{(i,j) \in F^*} c_{ij}
\end{array}
\]
Contradicting $F^*$ being the minimum weight perfect matching.
The reverse implication follows by similar argument in the reverse order.
\end{proof}

Using Lemma \ref{le.matching}, next we show that the dual variables of the $b$-matching LP give an equilibrium price vector.
In the $b$-matching LP there is an equality for each node to ensure its matching requirement. Let $u_v$, $u_i$ and $q_s$ be
the dual variables corresponding to the equalities of nodes $v$, $i \in A$ and $s \in S$. 
Then the dual LP for minimum weight perfect $b$-matching in graph $(A',S,E)$ is as follows. 

\begin{equation}
\label{eq.matching}
%
\begin{array}{ll}
\vspace{0.1cm}
\max  : & \displaystyle\sum_{i \in A} {u_i} + \displaystyle\sum_{s \in S} {cap(s)q_s} +u_v b_v\\
\vspace{0.1cm}
\mbox{s.t.}  & \forall i\in A,\ s \in W(i): u_i \le - q_s + c_{is} \\
 & \forall s \in S:\ \ \ \ \ \ \ \ \ \ \ \ \ \  u_v \le - q_s+1 
\end{array}
\end{equation}

There are no non-negativity constraints on the dual variables since the corresponding primal constraints are equality.

\begin{lemma}\label{le.optprices}
There exists a dual solution $(u^*,q^*)$ of (\ref{eq.matching}) with $u^*_v=1$, and given that, 
$-q^*$ yields a solution of LP (\ref{dual}). 
\end{lemma}
\begin{proof}
If $(u^*,q^*)$ is a dual solution then so is $v=(u^*+\delta,q^*-\delta)$ for any $\delta \in \mathbb R$. This is
because, clearly $v$ is feasible. Further, since $|A|+b_v=\sum_s cap(s)$ the value of objective function at $v$ is same as
that at $(u^*,q^*)$. 

Therefore given any solution of the dual, we can obtain one with $u^*_v=1$ by an additive scaling.
Replacing $u_v$ with $1$ and $q_s$ with $-p_s$ in (\ref{eq.matching}) gives $max\{\sum_i u_i - \sum_s cap(s) p_s + b_v\ |\ u_i \le
p_s+c_{is}, \ p_s \ge 0\}$, which is exactly (\ref{dual}), and hence the lemma follows.
%
\end{proof}

Since a primal and a dual solution of a minimum weight perfect $b$-matching can be computed in strongly
polynomial time \cite{schrijver.book}, the next theorem follows using Lemmas \ref{le.matching} and \ref{le.optprices}, and
Theorem \ref{thm.equi}.

\begin{theorem}
\label{thm.eqsp}
There is a combinatorial, strongly polynomial algorithm for computing an equilibrium landing schedule and 
equilibrium prices.
\end{theorem}
%
%

\subsection{Incentive Compatible in Dominant Strategy}
Since equilibrium price vectors of the market is in one-to-one correspondence with the solutions of the
dual matching LP with $u_v=1$ (Lemma \ref{le.optprices}), they need not be unique, and in fact form a convex set.
In this section we show that one of them induces VCG payments, and therefore is incentive compatible in dominant strategy.
Further, we will design a method to compute such VCG prices in strongly polynomial time. 

An instance of the perfect $b$-matching problem can be reduced to the perfect matching problem by duplicating node $n$,
$b_n$ times.  Therefore, if we convert the costs $c_{is}$ on edge $(i,s)$ to payoffs $H-c_{is}$ for a big enough constant
$H$, the market becomes an equivalent matching market (also known as assignment game) \cite{SS} where the costs of producing
goods, the slots in our case, are zero. It is not difficult to check that equilibrium allocations and prices of our original
market and the transformed matching market exactly match.

\rn{For such a market, Leonard \cite{leonard} showed that the set of equilibrium prices of a matching market \rn{with minimum sum} 
correspond precisely to VCG payments \cite{nisan.chap9}, thereby showing that the market is incentive compatible
in dominant strategy at such a price vector. Since the proof in \cite{leonard} is not formal, we have provided a complete proof in Appendix \ref{sec.incentive}. Since equilibrium prices form a lattice \cite{DG,SS,AJM}, the one minimizing sum has to be simultaneously minimum for all goods.\footnote{Equilibrium
prices $p$ are minimum if for any other equilibrium prices $p'$ we have $p_s \leq p'_s, \ \forall s \in S$.}}
Clearly, such a price vector has to be unique.  
Next we give a procedure to compute the minimum equilibrium price vector, starting from any equilibrium price vector $p^*$ and
corresponding equilibrium schedule $x^*$. 

The procedure is based on the following observation: Given equilibrium prices $p^*$ and corresponding schedule $x^*$, construct
graph $G(x^*,p^*)$ where slots form the node set. Put a directed edge from slot $s$ to slot $s'$ if there exists a flight, say $i$,
scheduled in $s$ at $x^*$, and it is indifferent between $s$ and $s'$ in terms of total cost, i.e. $x^*_{is}=1$ and
$p^*_s+c_{is}=p^*_{s'}+c_{is'}$. An edge in graph $G(x^*,p^*)$ indicates that if the price of slot $s'$ is decreased then
$i$ would prefer $s'$ over $s$. Therefore, in order to maintain $x^*$ as an equilibrium schedule the price of $s$ also has
to be decreased by the same amount. 

\begin{lemma}\label{le.path}
Prices $p^{*m}$ give the minimum equilibrium prices if and only if every node in $G(x^*,p^{*m})$ has a
directed path from a zero priced node, where $x^*$ is the corresponding equilibrium schedule.  
\end{lemma}
\begin{proof}
Suppose slot $s$ does not have a path from a zero priced node.
Consider the set $D$ of nodes which can reach $s$ in $G^*=G(x^*,p^{*m})$; clearly, they have positive prices.
Therefore, $\exists\epsilon>0$ such that the prices of all the slots in $D$ can
be lowered by $\epsilon$ without violating the equilibrium condition ($1$), contradicting minimality of $p^{*m}$.

For the other direction, the intuition is that if every node is connected to a zero priced node in $G(x^*,p^{*m})$, then price of any
slot can not be reduced without enforcing price of some other slot go negative, in order to get the corresponding equilibrium schedule.
The formal proof is as follows:

To the contrary suppose every node is connected to a zero priced node in $G^*$ and there are equilibrium
prices $p'\le p^{*m}$ such that for some $s \in S$, $p^{*m}_s > p'_s>0$. Consider, one such $s$ nearest to a zero-priced node in $G^*$.
Since, $p'_s\ge 0$, we have $p^{*m}_s>0$, and therefore $s$ is filled to its capacity at prices $p^{*m}$ (using equilibrium condition
(2) of Section \ref{sec.model}). Let $x'$ be
the equilibrium schedule corresponding to prices $p'$. 

Let $s' \rightarrow s$ in $G^*$. By choice of $s$ we have that $p'_{s'}=p^{*m}_{s'}$. In that case, a
flight, say $i'$, allocated to $s'$ at $p^*$ will move to $s$ at $p'$.
Implying that $\sum_i x'_{is'} < \sum_i x^*_{is'} \le cap(s')$. Hence $p'_{s'}=0\Rightarrow p^{*m}_{s'} =0$ (using equilibrium condition (2)). 
Let $Z=\{s\ |\ p^{*m}_s=0\}$. There are two cases at this point:
\medskip

\noindent{\bf Case I -} Flights in slot $s$ at $x^{*}$ remain in $s$ at $x'$, {\em i.e.}, $\{i\ | x^*_{is}=1\}\subseteq \{i\
| x'_{is}=1\}$:\\
Since, $x'_{i's}=1$ and $x^*_{i's}=0$, implying $\sum_i x'_{is} > \sum_i x^*_{is} = cap(s)$, a contradiction.
\medskip

\noindent{\bf Case II -} Some flight $i$ scheduled in $s$ at $x^*$, reschedules at $x'$, {\em i.e.,} $\exists i,\ x^*_{is}=1,
x'_{is}=0$:\\
Construct a graph $H$, where slots are nodes, and there is an edge from $u$ to $v$ if $\exists i,\ x^*_{iu}=1, x'_{iv}=1$, {\em
i.e.,} flight $i$ moved from $u$ to $v$ when prices are changed from $p^{*m}$ to $p'$, with weight being number of edges moved.
Note that price of every node with an incoming edge should have decreased while going from $p^{*m}$ to $p'$. 
Therefore, nodes of $Z$ have no incoming edges. Further, nodes with incoming edges are filled to capacity at $p^{*m}$ since their prices
are non-zero. In that case, total out going weight of such a node should be at least total incoming weight in $H$. 

If there is a cycle in $H$, then subtract weight of one from all its edges, and remove
zero-weight edges. Repeat this until there are no cycles. 
Since, $s' \in Z$, it had no incoming edge, but had an edge to $s$. Therefore, there is a path in remaining $H$ starting at $s'$. 
Consider the other end of this path. 
Clearly, it has to be filled beyond its capacity at $x'$, a contradiction.
\end{proof}

Using the fact established by Lemma \ref{le.path} next we design a procedure to compute the minimum equilibrium prices in
Table 1, given any equilibrium prices $p^*$ and corresponding schedule $x^*$. 

\begin{table}[!h]
\caption{Procedure for Computing Minimum Optimal Prices}
\label{alg.main}
\vspace{-0.4cm}

\begin{center}
\begin{tabular}{|l|}\hline
MinimumPrices($x^*,p^*$)\\
\hspace{5pt}$1.$ $Z\leftarrow$ Nodes reachable from zero-priced nodes in $G(x^*,p^*)$.\\ 
\hspace{5pt}$2.$ Pick a $d \in S\setminus Z$ \\
\hspace{5pt}$3.$ $D\leftarrow$ \{Nodes that can reach $d$ in $G(x^*,p^*)\}$, $\delta\leftarrow 0$,\\
\hspace{16pt}and $p^*_s\leftarrow p^*_s-\delta, \forall s \in D$ \\
\hspace{5pt}$4.$ Increase $\delta$ until one of the following happen \\
\hspace{16pt}- If price of a slot in $D$ becomes zero, then go to $1$. \\
\hspace{16pt}- If a new edge appears in $G(x^*,p^*)$, then recompute $Z$. \\
\hspace{20pt} If $d \in Z$ then go to $2$ else go to $3$. \\
\hspace{5pt}$5.$ Output $p^*$ as the minimum prices. \\ \hline
\end{tabular}
\end{center}
\vspace{-0.6cm}

\end{table}

\begin{lemma}\label{lem.min}
Given an equilibrium $(x^*,p^*)$, MinimumPrices($x^*,p^*$) outputs minimum prices in time $O(|A||S|^2)$.
\end{lemma}
\begin{proof}
Note that the size of $Z$ and edges in $G(x^*,p^*)$ are increasing. Therefore, Step $3$ is executed $O(|S|)$ many times in
total. Step $4$ may need $O(|A||S|)$ time to compute the threshold $\delta$. 
Therefore the running time of the procedure MinimumPrices is $O(|A||S|^2)$. 
Let the output price vector be $p^{*m}$.
The lemma follows from the fact that $(x^*,p^{*m})$ still satisfy both the equilibrium conditions, and every
slot is reachable from a zero priced node in $G(x^*,p^{*m})$ (Lemma \ref{le.path}).
\end{proof}

Theorems \ref{thm.equi} and \ref{thm.eqsp}, Lemma \ref{lem.min}, together with \cite{leonard} give: 

\begin{theorem}
\label{thm.main}
There exists an incentive compatible (in dominant strategy) market mechanism for scheduling a set of flight landings at a single
airport; moreover, it is computable combinatorially in strongly polynomial time.
\end{theorem}

\section{Dealing with Multiple Airports}
\label{sec.multiple}

In this section, we suggest how to use the above-stated solution to deal with unexpected events that result in global, cascading delays.
Our proposal is to decompose the problem of scheduling landing slots over a period of a day at
multiple airports into many small problems, each dealing with a set of flights whose arrival times lie in a window of a
couple of hours -- the window being chosen in such a way that all flights would already be in the air and their actual
arrival times, assuming no further delays, would be known to the airline companies and to FAA. At this point, an airline
company has much crucial information about all the other flights associated with its current flight due to connections, crew availability, etc.
It is therefore in a
good position to determine how much delay it needs to buy away for its flight and how much it is willing to pay, by setting
$c_{is}$s accordingly. This information is used by FAA to arrive at a landing schedule.  The process is repeated every
couple of hours at each airport.

\bibliographystyle{splncs03}
\bibliography{kelly}
\appendix
\section{Incentive Compatible in Dominant Strategy}
\label{sec.incentive}

In order to find a socially optimal and fair allocation, it is important that reporting true private information
be the best strategy for the agents; such a mechanism is called incentive compatible in dominant strategy.  
Hence an important question for our setting is: Can an airline secure a better deal by reporting fictitious $c_{is}$s? 

The dual (\ref{dual}) will have a convex set of equilibrium solutions (prices) in general. Among these, consider solutions that
minimize the sum of prices of all slots. Leonard \cite{leonard} showed that payment as per these prices, together with the optimal allocation, are incentive compatible in dominant strategy (DSIC). However the proof is in discussion form, and not absolutely formal. For convenience, we give a formal proof in this section. 

We will prove that equilibrium price vector with minimum sum is unique and moreover induces VCG payments. 
VCG payments together with social-welfare maximizing allocation induces truthful (DSIC) mechanism \cite{nisan.chap9}. 
First observe that such an equilibrium must have a slot with price zero -- otherwise
subtracting the minimum price from all slots leads to a better equilibrium.
Let $x^*$ and $p^*$ be equilibrium landing schedule and equilibrium prices respectively.
The optimality condition of Lemma \ref{lem.cost} can be rewritten as 

\begin{equation}\label{eq.opt}
x^*_{is}=1\ \ \ \Rightarrow\ \ \ p^*_s + c_{is} \le p^*_{is'} +
c_{is'},\ \ \forall s' \in W(i) 
\end{equation} 

Construct graph $G(x^*,p^*)$ where slots form the node set. Put a directed edge from slot $s$ to slot $s'$ with label $i$ if
there is a flight $i$ scheduled in $s$, and it is indifferent between $s$ and $s'$ in terms of total cost, i.e. $x^*_{is}=1$
and $p^*_s+c_{is}=p^*_{s'}+c_{is'}$. An edge in graph $G(x^*,p^*)$ indicates that if the price of slot $s'$ is decreased
then $i$ would prefer $s'$ over $s$ and hence violating condition (\ref{eq.opt}).

\begin{lemma}\label{le.path}
Let $p^*$ be equilibrium price minimizing $\sum_{s \in S} p_s$. Then every node in $G(x^*,p^*)$ has a
directed path from a zero priced node.  
\end{lemma}
\begin{proof}
Suppose slot $s$ does not have a path from a zero priced node.
Consider the set $D$ of nodes which can reach $s$ in $G(x^*,p^*)$; clearly, they have positive prices.
Therefore, $\exists\epsilon>0$ such that the prices of all the slots in $D$ can
be lowered by $\epsilon$ without violating the optimality condition (\ref{eq.opt}), contradicting minimality of $p^*$.
\end{proof}

\begin{lemma}
There is a unique equilibrium price vector that minimizes $\sum_s p_s$.
\end{lemma}
\begin{proof}
Suppose $p$ and $p'$ both are equilibrium price vectors minimizing $\sum_s p_s$. Then $\exists s \in S$ such that $p'_s <
p_s$. Consider the path $P$ in $G(x^*,p)$ from a zero price node, $s_0$, to $s$. 
Next consider prices $p'$. In order to satisfy 
condition \ref{eq.opt} for all the pairs of slots represented by edges on this path, the price of each node on this path
has to be reduced by $p_s - p'_s$, thereby assigning a negative price to $s_0$, contradicting the existence of $p'$.
\end{proof}

Henceforth let $p^*$ denote the equilibrium prices that minimize $\sum_{s \in S} p_s$.
Using the paths established in Lemma \ref{le.path}, next we give explicit expression for values of $p^*_s$ for all slots $s$.
Let the cost of a landing schedule $x$ be defined as $cost(x)=\sum_{i \in A, s \in W(i)} x_{is} c_{is}$. 
Now consider a path $P(s,s')$ from $s$ to $s'$ in $G(x^*,p^*)$. Let $x(P(s,s'))$ be a landing schedule same as $x^*$ except
that for each edge on path $P(s,s')$ the corresponding indifferent flight shifts to the slot the edge is going into. 
Note that $x(P(s,s'))$ may not be feasible in LP (\ref{LP}). In particular it may violate capacity constraint of slot $s'$.

\begin{lemma}\label{le.prices}
Let $P(s^0,s)$ be a directed path from $s^0$ to $s$ in $G(x^*,p^*)$, $p^*_{s^0}=0$, and $x'=x(P(s^0,s))$, then
$p^*_s=cost(x^*)-cost(x')$.
\end{lemma}
\begin{proof}
Let the path $P(s^0,s)$ be $s^0\xrightarrow{\ 1\ } s^1\xrightarrow{\ 2\ } \dots \xrightarrow{k-1} s^{k-1}
\xrightarrow{\ k\ } s^k=s$, where the labels above the arrows are indices of indifferent flights. Since $i$ is the
indifferent flight in slot $s^{i-1}$, we have $p^*_{s^i}=p^*_{s^{i-1}} + (c_{is^{i-1}}-c_{is^i}),\ 1<i\le k$. Putting all
these together we get $p^*_s=p^*_{s^0} + \sum_{i=1}^k (c_{is^{i-1}}-c_{is^i})$. 
Since $x'$ is obtained from $x^*$ by shifting flight $i$ from slot $s^{i-1}$ to $s^i$, we can rewrite this expression as 
$p^*_s=p^*_{s^0}+\sum_{i \in A, s\in S} c_{is}
(x^*_{is}-x'_{is})=p^*_{s^0}+cost(x^*)-cost(x')$. Since $p^*_{s^0}=0$, this proves the lemma.
\end{proof}

\begin{lemma}\label{le.vcgx}
Let $x^{vcg}$ be a VCG landing schedule, then $x^{vcg}$ is a solution of LP (\ref{LP}).
\end{lemma}
\begin{proof}
Since $c_{is}$ is the ``delay cost" incurred by flight $i$ if scheduled in slot $s$, $-c_{is}$ is the corresponding
``welfare". Therefore, social welfare maximizing feasible schedule has to minimize the overall delay costs of all the
flights subject to the capacity constraints of the slots. This is exactly same as solving LP (\ref{LP}).
\end{proof}

In order to calculate VCG payment for flight $i$ we need to find an optimal landing schedule when flight $i$ is not present. 
This can be obtained by solving LP (\ref{LP}) with flight set $A_{-i}=A\setminus {i}$. 
Let $x^{i}$ be a solution of this LP.

For a schedule $x$ and flight set $A'\subseteq A$, define $cost(x,A')$ as $\sum_{i \in A',s\in W(i)} x_{is} c_{is}$. 
The overall minimum delay cost of agents in $A_{-i}$ is $cost(x^*,A_{-i})$ when $i$ participates and it is
$cost(x^{i},A_{-i})$ when $i$ does not participate. Therefore, the VCG payment of flight $i$, say $VCGpay_i$, is
$cost(x^*,A_{-i})-cost(x^{i},A_{-i})$. Since the allocation $x^*_{A_{-i}}$ is feasible
and $x^{i}$ is optimal, we get $ cost(x^*,A_{-i})\ge cost(x^{i},A_{-i})\Rightarrow VCGpay_i\ge0$.

Define $H^i$ to be a graph with slots as nodes, and there is an edge from $s$ to $s'$ if for a flight $j$, landing slot
shifts from $s$ in $x^*$ to $s'$ in $x^{i}$. 

\begin{lemma}\label{le.vcgpath}
Let $s^i$ be the slot assigned to flight $i$ as per $x^*$. Then graph $H^i$ is a directed path with $s^i$ being the sink.
\end{lemma}
\begin{proof}
If there is no incoming edge to $s^i$ then there are no edges in $H^i$ or else $x^*$ can not be the optimal allocation for
set $A$.  This is because the shifting indicated by edges of $H^i$ is also possible without removing flight $i$, which
together reduces overall delay cost as $cost(x^{i},A_{-i})\le cost(x^*,A_{-i})$. By a similar argument it follows that
$H^i$ can not have any directed cycles. Therefore, $H^i$ is a directed acyclic graph, with at least one incoming edge
to $s^i$. 

Next we decompose the graph into two components, say $D$ and $F$, where $D$ is a directed path to $s^i$ from a
source, and $F$ is the rest of $H^i$. Note that once flight $i$ is removed, the shifting indicated by $D$ and $F$ can be
carried out separately while maintaining the capacity constraints of the slots. Therefore, both individually should improve
the delay cost. However, shifting of $F$ is possible even without removing flight $i$, contradicting $x^*$ being optimal
allocation for set $A$. Therefore, $H^i$ consists of only the path $D$. 
\end{proof}

Using the structure of graphs $H^i$ and $G(x^*,p^*)$, next we prove that VCG payments are same as the prices $p^*$
assigned by our market mechanism. 

\begin{lemma}\label{le.vcgpay}
For every flight $i\in A$, if $x^*_{is'}=1$ then $VCGpay_i=p^*_{s'}$.
\end{lemma}
\begin{proof}
Let $P(s^0,s')$ be a path from $s^0$ to $s'$ in $G(x^*,p^*)$ with
$p^*_{s^0}=0$, and let $x'=x(P(s^0,s'))$. Let $x''=x'_{A_{-i}}$.
Vector $x''$ is feasible in LP (\ref{LP}) with flight set $A_{-i}$. Since $x^{i}$ is the
optimal solution of this LP, we have
\[\hspace{0cm}
\begin{array}{l}
cost(x^{i},A_{-i}) \le cost(x'',A_{-i}) \\
\hspace{2cm}\begin{array}{ll}
\Rightarrow & cost(x^*,A_{-i})-cost(x^{i},A_{-i}) \ge
cost(x^*,A_{-i})-cost(x'',A_{-i}) \\
\Rightarrow & VCGpay_i \ge cost(x^*,A)-cost(x',A) \\
 & \hspace{2cm}(\because \forall s,\ x'_{is}=x^*_{is}, \mbox{ and } \ \
x'_{A_{-i}}=x'')\\
 \Rightarrow & VCGpay_i \ge p^*_{s'}\ \  (\mbox{Lemma \ref{le.prices}})
\end{array}
\end{array}
\]

Next we show that $p^*_{s'} \ge VCGpay_i$. 
Suppose the path $H^i$ is $y^0\rightarrow y^1 \rightarrow \dots \rightarrow y^{k-1} \rightarrow y^k$,
where $y^k=s'$ (Lemma \ref{le.vcgpath}), and flight $f_j$ shifts from slot $y^{j-1}$ in $x^*$ to slot $y^j$ in $x^{i}$.
Since, $x^*$ and $(t^*,p^*)$ are solutions of LP (\ref{LP}) and (\ref{dual}) respectively, using Lemma \ref{lem.cost} we have 
\[
p^*_{y^j} + c_{f_j y^j}
\ge p^*_{y^{j-1}} +c_{f_j y^{j-1}},\ 1\le j\le k
\]

Putting all of them together and canceling the intermediate price variables we have 
\[
\begin{array}{lcl}
p^*_{s'}=p^*_{y^k}& \ge& p^*_{y^0} + \sum_{j=1}^k (c_{f_j y^{j-1}} - c_{f_j y^j}) \\
& =& p^*_{y^0} + \sum_{j \in A_{-i}, s \in W(j)} c_{is} (x^*_{js}-x^{i}_{js})
\end{array}
\]
The last equality is due to the fact that $x^*$ and $x^{i}$ are same except for flights $f_j,\ 1 \le j \le k$ and flight
$i$. This gives $p^*_{s'}\ge p^*_{y^0} + cost(x^*,A_{-i})-cost(x^{i},A_{-i})$. Since, all the prices are non-negative, we
get $p^*_{s'} \ge VCGpay_i$.
\end{proof}

Lemmas \ref{le.vcgx} and \ref{le.vcgpay} together gives the following theorem.

\begin{theorem}\label{thm.ic}
Let $x^*$ be the equilibrium schedule and $p^*$ be the equilibrium prices where $\sum_{s \in S} p_s$ is minimum.
The mechanism which schedules flight $i$ to slot $s$ if $x^*_{is}=1$ and charges $p^*_s$ is incentive compatible in
dominant strategy. 
\end{theorem}

\end{document}